\documentclass[a4paper,UKenglish,cleveref, autoref, thm-restate]{lipics-v2021}

\usepackage{graphicx}
\usepackage[algo2e,ruled,lined]{algorithm2e}
\usepackage{color}
\usepackage{amssymb,amsmath,amsthm}
\usepackage{xcolor,pgf,tikz}
\usepackage{circuitikz}
\usepackage{mathtools}

\usetikzlibrary{arrows,automata,shapes,snakes,calc,positioning,petri}
\usepackage{tcolorbox}
\usepackage{soul}
\usepackage[shortlabels]{enumitem}

\usepackage{array}
\usepackage{booktabs}

\usepackage{thmtools}
\usepackage{thm-restate}

\usepackage{comment}

\newcommand{\nat}{\mathbb{N}}
\newcommand{\integer}{\mathbb{Z}}
\newcommand{\rat}{\mathbb{Q}}
\newcommand{\real}{\mathbb{R}}

\newcommand{\Front}{{\sf Front}} 

\newcommand{\prob}{{\bf Pr}}

\newcommand{\suiv}{p}

\newcommand{\alain}[1]{\textcolor{blue}{#1}}

\nolinenumbers

\title{About Decisiveness of Dynamic\\ Probabilistic Models} 
\titlerunning{About Decisiveness of Dynamic Probabilistic Models}

\author{Alain {Finkel}}{Universit\'e Paris-Saclay, CNRS, 
ENS Paris-Saclay, IUF, Laboratoire Méthodes Formelles,  
91190, Gif-sur-Yvette, France \and \url{https://ens-paris-saclay.fr/alain-finkel/} }{finkel@lsv.fr}{https://orcid.org/0000-0003-2482-6141}{His work has been supported by ANR project BRAVAS (ANR-17-CE40-0028).}

\author{Serge {Haddad}}{Universit\'e Paris-Saclay,  CNRS, 
ENS Paris-Saclay, Laboratoire Méthodes Formelles,  INRIA
91190, Gif-sur-Yvette, France \and \url{http://www.lsv.fr/~haddad/} }{haddad@lsv.fr}{https://orcid.org/0000-0002-1759-1201}{His work has been supported by ANR project MAVeriQ (ANR-20-CE25-0012).}

\author{Lina {Ye}}{Universit\'e Paris-Saclay, CNRS, 
ENS Paris-Saclay, CentraleSupélec, Laboratoire Méthodes Formelles,  
91190, Gif-sur-Yvette, France \and \url{https://www.lri.fr/~linaye/} }{lina.ye@universite-paris-saclay.fr}{https://orcid.org/0000-0002-2217-4752}{}

\authorrunning{A. Finkel et al.} 

\Copyright{Alain Finkel,  Serge Haddad and Lina Ye} 



\category{} 

\relatedversion{} 

\begin{document}

\maketitle

\begin{abstract}
Decisiveness of infinite Markov chains with respect to some (finite or infinite) target set of states 
is a key property that allows to compute the reachability probability of this set up to an arbitrary precision. 
Most of the existing works assume constant weights for defining the probability of a transition in the considered models.
However numerous  probabilistic modelings require (dynamic) weights that depend on both the current state and the transition.
So we introduce a dynamic probabilistic version
of counter machine (pCM). After establishing that decisiveness is undecidable for pCMs even with constant weights,
we study the decidability of decisiveness
for subclasses of pCM. We show that, without restrictions on dynamic weights, decisiveness is undecidable
with a single state and single counter pCM. On the contrary with polynomial weights, decisiveness becomes decidable
for single counter pCMs under mild conditions. Then we show that  decisiveness of probabilistic Petri nets  (pPNs) with polynomial weights
is undecidable even when the target set is upward-closed unlike the case of constant weights. Finally we prove that
the standard subclass of pPNs with a regular language is decisive with respect to a finite set whatever the kind of weights.
\end{abstract}

\section{Introduction}
\label{sec:introduction}

\noindent
{\bf Probabilistic models.}
Since the 1980's, finite-state Markov chains have been considered for the modeling and analysis 
of probabilistic concurrent finite-state programs~\cite{DBLP:conf/focs/Vardi85}. 
More recently this approach has been extended to the verification of the infinite-state Markov chains 
obtained from probabilistic versions 
of automata extended with  unbounded data (like stacks, channels, counters, clocks). 
For instance in 1997,  \cite{DBLP:conf/tapsoft/IyerN97} started the study of \emph{probabilistic lossy channel systems} (pLCS).
In 2004, the model checking of \emph{probabilistic pushdown automata} (pPDA) 
appeared in many works~\cite{DBLP:conf/lics/EsparzaKM04,DBLP:conf/lics/EsparzaKM05,DBLP:conf/focs/BrazdilEK05,Esparza06} (see \cite{DBLP:journals/fmsd/BrazdilEKK13} for a survey).
\emph{Probabilistic counter machines} (pCM) have also been studied since 2010~\cite{DBLP:journals/jacm/BrazdilKK14}.

\smallskip\noindent
{\bf Computing the probability of reachability.} Here we focus on the problem of \emph{Computing the Reachability Probability up to an arbitrary precision} (CRP). 
There are (at least) two strategies to solve the CRP problem. 

The first one is to consider the Markov chains associated with a particular class of probabilistic models (like pPDA or pPN)
and some specific target sets 
and to exploit the properties of these models to design a CRP-algorithm.
For instance in~\cite{DBLP:journals/fmsd/BrazdilEKK13}, the authors 
exhibit a PSPACE algorithm for pPDA and PTIME algorithms for single-state pPDA and for one-counter automata.

The second one consists in exhibiting a property of Markov chains that yields a generic algorithm
for solving the CRP problem and then looking for models that generate Markov chains that fulfill this property.
\emph{Decisiveness} of Markov chains is such a property and it has been shown that
 pLCS are decisive and that probabilistic Petri nets (pPN) are decisive when the target set is upward-closed~\cite{AbdullaHM07}. 

\smallskip\noindent
{\bf Two limits of the previous approaches.} 
In most of the works, the probabilistic models associate a constant (also called \emph{static}) weight for transitions
and get transition probabilities by normalizing these weights among the enabled transitions in the current state
(except for some semantics of pLCS like in~\cite{DBLP:conf/tapsoft/IyerN97} where transition probabilities depend on the state due to the possibility of message losses). 
This forbids to model phenomena like congestion in networks (resp. performance collapsing in distributed systems) 
when the number of messages (resp. processes) exceeds some threshold leading to an increasing probability of message arrivals (resp. process creations)
before message departures (resp. process terminations). In order to handle them, one needs to consider 
 \emph{dynamic} weights i.e., weights depending on the current state.
 
 Generally given some probabilistic model and some kind of target set of states, 
 it may occur that  some instances of the model are decisive and some others are not.
 This raises the issue of the decidability status of the decisiveness problem.
 Interestingly, the decidability of the decisiveness property has only be studied and shown decidable for pPDA with constant weights~\cite{Esparza06}.


\smallskip\noindent
{\bf Our contributions.}
\begin{itemize}
	\item In order to unify our analysis of decisiveness, we introduce a dynamic probabilistic version of counter machine (pCM)
	and we first establish that decisiveness is undecidable for pCMs even with constant weights.
	\item Then we study the decidability of decisiveness of one-counter pCMS.
        We show that, without restrictions on dynamic weights, decisiveness is undecidable
        for one-counter pCM even with a single state. On the contrary, with polynomial weights, decisiveness becomes decidable
        for a large subclass of one-counter pCMs, called homogeneous probabilistic counter machine (pHM). 
	\item   Then we show that  decisiveness of probabilistic Petri nets  (pPNs) with polynomial weights is undecidable  when the target set is finite or upward-closed (unlike the case of constant weights). 
	Finally we prove that the standard subclass of pPNs with a regular language is decisive with respect to a finite set whatever the kind of weights.

%
	\item Some of our results are not only technically involved but contain new ideas. In particular, 
	the proof of undecidability of decisiveness for pPN with polynomial weights with respect to a finite or upward closed set is based on an original weak simulation of CM. 
	Similarly the model of pHM can be viewed as a dynamic extension of quasi-birth–death processes well-known in the performance evaluation field~\cite{1034625256}.
%
%
\end{itemize}

\smallskip\noindent
{\bf Organisation.} Section~\ref{sec:decisiveness} recalls decisive Markov chains, 
presents the classical algorithm for solving the CRP problem 
and shows that decisiveness is somehow related to recurrence of Markov chains. 
In section~\ref{sec:pCM},  we introduce pCM and show that decisiveness is undecidable for static pCM.  
In section~\ref{sec:oneCM}, we study the decidability status of decisiveness for probabilistic one-counter pCM
and in section~\ref{sec:pPN}, the decidability status of decisiveness for pPN.
Finally in Section~\ref{sec:conclusion}  we conclude and give some perspectives to this work.
All missing proofs can be found in Appendix.

\section{Decisive  Markov chains}
\label{sec:decisiveness}

As usual, $\nat$ and $\nat^*$ denote respectively the set of non negative integers and the set of positive integers.
The notations 
$\rat$, $\rat_{\geq 0}$ and $\rat_{>0}$ denote the set of 
 rationals, non-negative rationals and positive rationals.
Let $F\subseteq E$. When there is no ambiguity about $E$, $\overline{F}$ will denote $E\setminus F$.
\subsection{Markov chains: definitions and properties}

\noindent
{\bf Notations.} 
A set $S$ is \emph{countable} if there exists an injective function from $S$ 
to  the set of natural numbers: hence it could be finite or countably infinite.
Let $S$ be a countable set of elements called states.
Then $Dist(S) = \{\Delta : S \rightarrow \rat_{\geq 0} \mid \sum_{s\in S}\Delta(s)=1\}$  is the set of \emph{rational distributions} over $S$.
Let $\Delta \in Dist(S)$, then $Supp(\Delta)=\Delta^{-1}(\rat_{>0})$.
Let $T\subseteq S$, then  $S\setminus T$ will also be denoted $\overline{T}$.

\begin{definition}[(Effective) Markov chain]
A \emph{Markov chain} $\mathcal M=(S,\suiv)$ is a tuple where:
\begin{itemize}
	\item  $S$ is a countable set of states,
	\item $\suiv$ is the transition function from $S$ to $Dist(S)$;
%
%
\end{itemize}
When for all $s\in S$, $Supp(\suiv(s))$ is finite and the function $s \mapsto \suiv(s)$ is computable,
one says that  $\mathcal M$ is \emph{effective}.
\end{definition}
When $S$ is countably infinite, we say that $\mathcal M$ is \emph{infinite} and we sometimes identify $S$ with $\nat$.
We also denote $\suiv(s)(s')$ by $\suiv(s,s')$
and $\suiv(s,s')>0$ by
$s\xrightarrow{\suiv(s,s')} s'$.  A Markov chain is also viewed as a transition system
whose transition relation $\rightarrow$ is defined by $s\rightarrow s'$ if $\suiv(s,s')>0$. 

%
%

\begin{example}
Let $\mathcal M_1$ be the Markov chain of Figure~\ref{fig:rw}.
In any state $i>0$, the probability for going to the ``right'', $p(i,i+1)=\frac{f(i)}{f(i)+g(i)}$ and  for going to the ``left'', $p(i,i-1)=\frac{g(i)}{f(i)+g(i)}$.
In state $0$, one goes to  $1$ with probability 1. $\mathcal M_1$ is effective if the functions $f$ and $g$ are
computable.
\end{example}

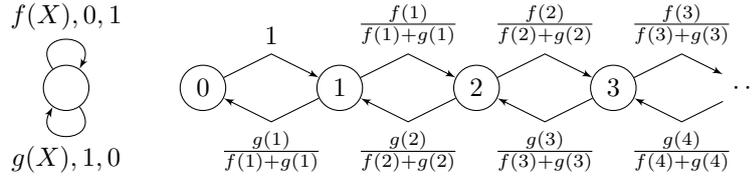
\begin{figure}
\begin{center}
  \begin{tikzpicture}[node distance=2cm,->,auto,-latex,scale=0.9]
     \path (-2,0) node[minimum size=0.6cm,draw,circle,inner sep=2pt] (q5) {}; 
    
   \path (0,0) node[minimum size=0.6cm,draw,circle,inner sep=2pt] (q0) {$0$};
   \path (2,0) node[minimum size=0.6cm,draw,circle,inner sep=2pt] (q1) {$1$};
   \path (4,0) node[minimum size=0.6cm,draw,circle,inner sep=2pt] (q2) {$2$};
   \path (6,0) node[minimum size=0.6cm,draw,circle,inner sep=2pt] (q3) {$3$};
   \path (8,0) node[] (q4) {$\cdots$};

      

    \draw[arrows=-latex'] (q0)-- (1,0.5) node[pos=1,above] {$1$}--(q1) ;
    \draw[arrows=-latex'] (q1)-- (1,-0.5) node[pos=1,below] {$\frac{g(1)}{f(1)+g(1)}$}--(q0) ;

    \draw[arrows=-latex'] (q1)-- (3,0.5) node[pos=1,above] {$\frac{f(1)}{f(1)+g(1)}$}--(q2) ;
    \draw[arrows=-latex'] (q2)-- (3,-0.5) node[pos=1,below] {$\frac{g(2)}{f(2)+g(2)}$}--(q1) ;

    \draw[arrows=-latex'] (q2)-- (5,0.5) node[pos=1,above] {$\frac{f(2)}{f(2)+g(2)}$}--(q3) ;
    \draw[arrows=-latex'] (q3)-- (5,-0.5) node[pos=1,below] {$\frac{g(3)}{f(3)+g(3)}$}--(q2) ;

    \draw[arrows=-latex'] (q3)-- (7,0.5) node[pos=1,above] {$\frac{f(3)}{f(3)+g(3)}$}--(q4) ;
    \draw[arrows=-latex'] (q4)-- (7,-0.5) node[pos=1,below] {$\frac{g(4)}{f(4)+g(4)}$}--(q3) ;

   \draw[-latex'] (q5) .. controls +(125:30pt) and +(55:30pt) .. (q5) node[pos=.5,above] {$f(X),0,1$};
   \draw[-latex'] (q5) .. controls +(-55:30pt) and +(-125:30pt) .. (q5) node[pos=.5,below] {$g(X),1,0$};

  \end{tikzpicture}
\end{center}
\caption{A pCM and its Markov chain $\mathcal M_1$ with for all $n\in \nat$, $0<f(n)$ and $0<g(n)$.}
\label{fig:rw}
\end{figure}
We denote $\rightarrow^*$,
the reflexive and transitive closure of $\rightarrow$ and we say that
$s'$ is  \emph{reachable from $s$} if  $s \rightarrow^* s'$. 
We say that a subset $A \subseteq S$ is \emph{reachable} from $s$ if some $s'\in A$ is reachable from $s$ and we denote $s \rightarrow^*A$.
Let us remark that every finite path of 
$\mathcal M$ can be extended into (at least) one infinite path.

\smallskip
Given an initial state $s_0$, the \emph{sampling} of a Markov chain $\mathcal M$ is an \emph{infinite
random sequence of states} (i.e., a path) $\sigma=s_0s_1\ldots$ such that for all $i\geq 0$,
$s_i\rightarrow s_{i+1}$. As usual,  the corresponding $\sigma$-algebra
is generated by the finite prefixes of infinite paths and the probability of a measurable subset  $\Pi$  of infinite paths, given an initial state $s_0$, is denoted $\prob_{\mathcal M,s_0}(\Pi)$. In particular denoting $s_0\ldots s_nS^\omega$ the set of infinite paths with  $s_0\ldots s_n$ as prefix
$\prob_{\mathcal M,s_0}(s_0\ldots s_nS^\omega)=\prod_{0 \leq i<n} \suiv(s_i,s_{i+1})$.
%

\smallskip
\noindent {\bf Notations.}
From now on,  {\bf G} (resp. {\bf F}, {\bf X}) denotes the always (resp. eventual, next)
operator of LTL,  and {\bf E} the existential operator of CTL$^*$~\cite{ModelChecking08}. 

\smallskip
Let $A\subseteq S$. We say that $\sigma$ 
\emph{reaches} $A$ if $\exists i\in \nat\ s_i\in A$ and that  $\sigma$ 
\emph{visits} $A$ if $\exists i>0\ s_i\in A$. 
The probability that starting from $s_0$, the path $\sigma$ reaches (resp. visits) $A$ will be denoted
by $\prob_{\mathcal M,s_0}({\bf F}  A)$ (resp. $\prob_{\mathcal M,s_0}({\bf XF}  A)$).

\smallskip
The  next definition states qualitative and quantitative properties of a Markov chain.
\begin{definition}[Irreducibility, recurrence, transience] Let $\mathcal M=(S,\suiv)$ be a Markov chain and  $s \in S$. Then:
\begin{itemize}
	\item $\mathcal M$ is \emph{irreducible} if for all $s,s'\in S$, $s\rightarrow^* s'$;
	\item $s$ is \emph{recurrent} if $\prob_{\mathcal M,s}({\bf XF}  \{s\})=1$
	otherwise $s$ is \emph{transient}.
%
\end{itemize}
\end{definition}

The next proposition states that in an irreducible Markov chain, all states are in the same category~\cite{KSK76}.
\begin{proposition} Let $\mathcal M=(S,\suiv)$ be an  irreducible Markov chain and $s,s' \in S$. Then
$s$ is recurrent if and only if $s'$ is recurrent.
\end{proposition}
Thus an irreducible Markov chain will be said transient or 
recurrent depending on the category of its states (all states are in the same category). 
In the remainder of this section,
we will relate this category with techniques for computing reachability
probabilities.

\begin{example}
$\mathcal M_1$ of Figure~\ref{fig:rw} is clearly irreducible. Let us define $p_n=\frac{f(n)}{f(n)+g(n)}$. Then (see Proposition \ref{appendix} in the Appendix), $\mathcal M_1$ is  recurrent if and only if  
$\sum_{n\in \nat} \prod_{1\leq m<n} \rho_m =\infty$ with $\rho_m=\frac{1-p_m}{p_m}$, and when transient, the probability that starting from $i$ the random path reaches
$0$ is equal to $\frac{\sum_{i\leq n} \prod_{1\leq m<n} \rho_m}{\sum_{n\in \nat} \prod_{1\leq m<n} \rho_m}$.

\end{example}

\subsection{Decisive Markov chains}

%
One of the goals of the quantitative analysis of  infinite Markov 
chains is to approximately compute reachability probabilities. 
Let us formalize it. Given a finite representation 
of a subset $A\subseteq S$,
one says that this representation is \emph{effective} if one can decide the membership problem for $A$. 
With a slight abuse of language, we
identify $A$ with any effective representation of $A$.


\smallskip
\centerline{ {\bf \small{The Computing of Reachability Probability  (CRP) problem}} }
\begin{center}
\fbox{
\begin{minipage}{0.95\textwidth}

\noindent
$\bullet$ Input: an effective Markov chain $\mathcal M$, an (initial) state $s_0$, 
an effective subset of states $A$, and a rational $\theta>0$.

\noindent
$\bullet$ 
 Output: an interval $[low,up]$ such that $up-low\leq \theta$ and $\prob_{\mathcal M,s_0}({\bf F}  A) \in [low,up]$.
\end{minipage}
}
\end{center}


In finite Markov chains, there is a well-known algorithm for computing exactly the reachability probabilities in
 polynomial time~\cite{ModelChecking08}. In infinite Markov chains, there are (at least) two possible research directions:
(1) either using the specific features of a formalism to design such a CRP algorithm~\cite{Esparza06}, (2) or requiring a
supplementary property on Markov chains in order to design an ``abstract'' algorithm, then verifying that given
a formalism this property is satisfied and finally transforming this algorithm into a concrete one.
\emph{Decisiveness}-based approach follows the second direction~\cite{AbdullaHM07}.
In words, decisiveness w.r.t. $s_0$ and $A$ means that almost surely the random path $\sigma$
starting from $s_0$ will reach $A$ or some state $s'$ from which $A$ is unreachable.
%
\begin{definition}  A Markov chain $\mathcal M$ is \emph{decisive} w.r.t.  $s_0\in S$ and $A\subseteq S$ if:
$$\prob_{\mathcal M,s_0}({\bf G}  (\overline{A}\cap {\bf EF} A))=0$$
\end{definition}
%
%
%
Then under the hypotheses of decisiveness w.r.t. $s_0$ and $A$ and decidability of the reachability problem w.r.t. $A$, 
Algorithm~\ref{algo:prob-reach-dec} solves the CRP problem.

Let us explain  Algorithm~\ref{algo:prob-reach-dec}. If $A$ is unreachable from $s_0$ it returns the singleton interval $[0,0]$.
Otherwise it maintains a lower bound $pmin$ (initially 0) and an upper bound $pmax$ (initially 1)
of the reachability probability and builds some prefix of the infinite execution tree of $\mathcal M$. It also maintains the probability to reach
a vertex in this tree. There are three possible cases when examining the state $s$ associated with the current vertex along a path of probability $q$:
(1) either $s\in A$ and the lower bound is incremented by  $q$,
(2) either $A$ is unreachable from $s$ and the upper bound is decremented by $q$,
(3) or it extends the prefix of the tree by the successors of $s$.
The lower bound always converges to the searched probability while due to the decisiveness property, the upper bound also converges to it
ensuring termination of the algorithm.
%
\begin{proposition}[\!\!\cite{AbdullaHM07}]\label{prop:prob-reach-dec}
Algorithm~\ref{algo:prob-reach-dec} terminates and computes an interval 
of length at most $\theta$ containing $\prob_{\mathcal M,s_0}({\bf F} A)$ when applied to a decisive Markov chain $\mathcal M$
w.r.t. $s_0$ and $A$ with a decidable reachability problem w.r.t. $A$.
\end{proposition}

Algorithm~\ref{algo:prob-reach-dec} can be applied to probabilistic Lossy Channel Systems (pLCS) since they are decisive (Corollary 4.7 in \cite{AbdullaHM07} and see \cite{AbdullaR03} for the first statement) and reachability is decidable in LCS \cite{AbdullaJ96}.  It can be also applied to pVASSs w.r.t. upward closed sets
 because Corollary 4.4 in \cite{AbdullaHM07} states that  pVASSs are decisive w.r.t. \emph{upward closed sets}. Observe that these results hold due to restrictions on transition probabilities that we will discuss later on.

\noindent
{\bf Observations.} The test $pmin =0$ is not necessary but adding it avoids to return 0 as lower
bound, which would be inaccurate since entering this loop means that $A$ is reachable from $s_0$.
Extractions from the front are performed in a way that the execution tree will be covered (for instance by a breadth first exploration). 

\begin{algorithm2e}
  \LinesNumbered
  \DontPrintSemicolon
  \SetKwFunction{CompProb}{CompProb}
  \SetKwFunction{PosCov}{PosCov}
  \SetKwFunction{Insert}{Insert}
  \SetKwFunction{Extract}{Extract}

{\CompProb}$(\mathcal M,s_0,A,\theta)$\;

 \lIf {{\bf not} $s_0 \rightarrow^* A$}{ \Return$(0,0)$} 
$pmin \leftarrow 0$; $pmax \leftarrow 1$; $\Front \leftarrow \emptyset$\;
 $\Insert(\Front,(s_0,1))$\;
 \While{$pmax-pmin >\theta$ {\bf or} $pmin =0$}
 { 
    $(s,q) \leftarrow \Extract(\Front)$\;
    \lIf {$s\in A$}{$pmin \leftarrow pmin+q$} 
    \lElseIf  {{\bf not}  $s \rightarrow^* A$}{$pmax \leftarrow pmax-q$} 
    \Else
     {
         \For {$s'\! \in Supp(\suiv(s))$} {$\Insert(\Front,(s',q \suiv(s,s'))$} 
      }
 }   
 \Return$(pmin,pmax)$

\caption{Framing the reachability probability in decisive Markov chains}
 \label{algo:prob-reach-dec}
 \end{algorithm2e}

\noindent

Let $\mathcal M$ be a Markov chain. 
One denotes $Post_{\mathcal M}^*(A)$, the set of states
that can  be reached from some state of $A$ and $Pre_{\mathcal M}^*(A)$, the set of states
that can reach $A$. 
%
%
While decisiveness has been used in several contexts including uncountable probabilistic systems~\cite{BertrandBBC18},
its relation with standard properties of Markov chains has not been investigated. This is the goal
of the next definition and proposition.

\begin{definition}
Let $\mathcal M$ be a Markov chain, $s_0\in S$ and $A\subseteq S$ such that $s_0 \not\in A\cup \overline{Pre_{\mathcal M}^*(A)}$.
The Markov chain $\mathcal M_{s_0,A}=(S_{s_0,A},\suiv_{s_0,A})$ is defined as follows:
\begin{itemize}
	\item $S_{s_0,A}$ is the union of (1) the smallest set containing $s_0$ and such that for all $s\in S_{s_0,A}$ and $s' \not\in A\cup \overline{Pre_{\mathcal M}^*(A)}$
	with $s\rightarrow s'$, one have:\\ $s'\in S_{s_0,A}$
	and (2) $\{s_\bot\}$ where $s_\bot$ is a new state; 
	\item for all $s,s'\neq s_\bot$, $\suiv_{s_0,A}(s,s')=\suiv(s,s')$
	and $\suiv_{s_0,A}(s,s_\bot)=\sum_{s' \in A\cup \overline{Pre_{\mathcal M}^*(A)}}\suiv(s,s')$; 
	\item $\suiv_{s_0,A}(s_\bot,s_0)=1$.
\end{itemize}
\end{definition}
The equivalence between decisiveness of $\mathcal M$ w.r.t. $s_0\in S$ and $A\subseteq S$ and recurrence of  $\mathcal M_{s_0,A}$ allows to apply standard
criteria for recurrence in order to check decisiveness. For instance we will use the criterion of the Markov chain of Figure~\ref{fig:rw} in Section~\ref{sec:oneCM}.

\begin{proposition}\label{rec-dec-rec}
Let $\mathcal M=(S,\suiv)$ be a Markov chain, $s_0\in S$ and $A\subseteq S$ such that $s_0 \not\in A\cup \overline{Pre_{\mathcal M}^*(A)}$. Then $\mathcal M_{s_0,A}$
is irreducible.
Furthermore  $\mathcal M$ is decisive w.r.t. $s_0$ and $A$ if and only if $\mathcal M_{s_0,A}$ is recurrent.
\end{proposition}
\begin{proof}
Let $s\in S_{s_0,A} \setminus \{s_\bot\}$. Then $s$ is reachable from $s_0$ and $A$ is reachable from $s$ in $\mathcal M$
implying that $s\rightarrow^*s_\bot$ in  $\mathcal M_{s_0,A}$ (using a shortest path for reachability). Since
$s_\bot \rightarrow s_0$, $s_\bot \rightarrow^* s$. Thus $\mathcal M_{s_0,A}$ is irreducible.

\noindent
$\mathcal M_{s_0,A}$ is recurrent iff $\prob_{\mathcal M_{s_0,A},s_\bot}({\bf XF}\{s_\bot\})=1$ iff 
$\prob_{\mathcal M_{s_0,A},s_0}({\bf F}\{s_\bot\})=1$ iff
$\prob_{\mathcal M,s_0}({\bf F}A\cup \overline{Pre_{\mathcal M}^*(A)})=1$ iff 
$\mathcal M$ is decisive w.r.t. $s_0$ and $A$.
%
%
\end{proof}

\section{Probabilistic counter machines}
\label{sec:pCM}

We now introduce \emph{probabilistic Counter Machines (pCM)} in order to study
the decidability of the decisiveness property w.r.t. several relevant subclasses
of pCM.


\begin{definition}[pCM]
A \emph{probabilistic counter machine (pCM)} is a  tuple $\mathcal C= (Q,P,\Delta,W)$ 
where:
\begin{itemize}
	\item $Q$ is a finite set of control states;
	\item $P=\{p_1,\ldots,p_d\}$ is a finite set of counters (also called places);
	\item $\Delta=\Delta_0 \uplus \Delta_1$ 
	where $\Delta_0$ is a finite subset of $Q \times P  \times \nat^d \times Q$\\
	and $\Delta_1$ is a finite subset of $Q \times \nat^d\times \nat^d \times Q$;
	\item $W$ is a computable function from $\Delta  \times \nat^d$ to $\nat^*$.
%
\end{itemize}
\end{definition}

\noindent
{\bf Notations.} A transition $t\in \Delta_0$ is denoted $t=(q^-_t,p_t,\mathbf{Post}(t),q^+_t)$ and also $q^-_t\xrightarrow{p_t,\mathbf{Post}(t)} q^+_t$.
 A transition $t\in \Delta_1$ is denoted $t=(q^-_t,\mathbf{Pre}(t),\mathbf{Post}(t),q^+_t)$
 and also $q^-_t\xrightarrow{\mathbf{Pre}(t),\mathbf{Post}(t)} q^+_t$.
 Let $t$ be a transition of $\mathcal C$. Then $W(t)$ is the function from $\nat^d$ to $\rat_{>0}$
defined by $W(t)(\mathbf{m})=W(t,\mathbf{m})$.  A polynomial is  \emph{positive}
if all its coefficients are non-negative 
and there is a positive constant term.
When for all $t\in T$, $W(t)$ is a positive polynomial whose variables are the counters,
we say that  $\mathcal C$ is a \emph{polynomial} pCM.

\smallskip
A \emph{configuration} of $\mathcal C$ is an item of $Q\times \nat^d$. Let $s=(q,\mathbf{m})$ be a configuration and
$t=(q^-_t,p_t,\mathbf{Post}(t),q^+_t)$ be a transition in $\Delta_0$. Then $t$ is \emph{enabled} in $s$ if $\mathbf{m}(p_t)=0$ and $q=q^-_t$; its \emph{firing}
 leads to the configuration $(q^+_t,\mathbf{m}+\mathbf{Post}(t))$.
Let $t=(q^-_t,\mathbf{Pre}(t),\mathbf{Post}(t),q^+_t)\in \Delta_1$. Then $t$ is \emph{enabled} in $s$ if $\mathbf{m}\geq \mathbf{Pre}(t)$ and $q=q^-_t$;  its \emph{firing}
leads to the configuration $s'=(q^+_t,\mathbf{m}-\mathbf{Pre}(t)+\mathbf{Post}(t))$. One denotes the configuration change by:
$s \xrightarrow{t} s'$. One denotes $En(s)$, the set of transitions enabled in $s$ and $Weight(s)=\sum_{t \in En(s)} W(t,\mathbf{m})$. 
Let $\sigma=t_1\ldots t_n$
be a sequence of transitions. We define the enabling and the firing of $\sigma$ by induction.
The empty sequence is always enabled in   $s$ and its firing leads to $s$.
When $n>0$,   $\sigma$ is enabled if $s\xrightarrow{t_1}s_1$
and $t_2  \ldots t_n$ is enabled in $s_1$. The firing of $\sigma$ leads to the configuration
reached by $t_2  \ldots t_n$ from $s_1$. A configuration $s$ is \emph{reachable} from some $s_0$
if there is a firing sequence $\sigma$ that reaches $s$ from $s_0$.  When $Q$ is a singleton, one omits the control states in the definition of transitions and configurations.

%
%

\smallskip
We now provide the semantic of a pCM as a countable Markov chain.

\begin{definition} Let $\mathcal C$ be a pCM. Then the Markov chain $\mathcal M_\mathcal C=(S,p)$ is defined by:
\begin{itemize}
	\item $S=Q\times \nat^d$;
	\item For all $s=(q,\mathbf{m})\in S$, if $En(s)=\emptyset$ then $p(s,s)=1$.
	Otherwise for all $s'\in S$: $$p(s,s')=Weight(s)^{-1}\sum_{s \xrightarrow{t} s'} W(t,\mathbf{m})$$
\end{itemize}

\end{definition}

\noindent
{\bf Notation.} Let $s\in S$.
For sake of clarity, $\prob_{\mathcal C,s}$ will denote $\prob_{\mathcal M_\mathcal C,s}$.

For establishing the undecidability results, we will reduce an undecidable problem related to \emph{counter programs}, which are a variant of CM.
Let us recall that a \emph{$d$-counter program} $\mathcal P$ is defined by a set of $d$ counters $\{c_1,\ldots,c_d\}$ and a set of $n+1$ instructions labelled by $\{0,\ldots ,n\}$, 
where for all $i<n$, the instruction $i$ is of type
\begin{itemize}
	\item either (1)  $c_j \leftarrow c_j+1; \mathbf{ goto~} i'$ 
	with $1\leq j\leq d$ and $0\leq i' \leq n$  
	\item or (2) $\mathbf{if~} c_j>0 \mathbf{~then~} c_j \leftarrow c_j-1; \mathbf{goto~} i'$, $\mathbf{else\ goto~} i''$
	with $1\leq j\leq d$ and $0\leq i',i'' \leq n$ 
\end{itemize}
and the instruction
$n$ is $\mathbf{halt}$. The program starts at instruction $0$ and halts if it reaches the instruction $n$.

The halting problem for  two-counter programs asks, given a two-counter program $\mathcal P$ and initial values of counters, 
whether $\mathcal P$ eventually halts. It is undecidable~\cite{Minsky67}.
We introduce a subclass of  two-counter programs that we call \emph{normalized}.
A normalized two-counter program $\mathcal P$ starts by resetting its counters and, on termination, resets its counters before halting.

\noindent
{\bf Normalized two-counter program.}
The first two 
instructions of a normalized two-counter program reset counters $c_1,c_2$ as follows:
\begin{itemize}[nosep]
	\item$0: \mathbf{~if~} c_1>0 \mathbf{~then~} c_1 \leftarrow c_1-1; \mathbf{goto~} 0$ $\mathbf{else\ goto~} 1$
	\item$1: \mathbf{~if~} c_2>0 \mathbf{~then~} c_2 \leftarrow c_2-1; \mathbf{goto~} 1$ $\mathbf{else\ goto~} 2$
\end{itemize}
The last three 
instructions of a normalized  two-counter program are:
\begin{itemize}[nosep]
	\item$n\!-\!2: \mathbf{~if~} c_1>0 \mathbf{~then~} c_1 \leftarrow c_1-1; \mathbf{goto~} n\!-\!2$ $\mathbf{else\ goto~} n\!-\!1$
	\item$n\!-\!1: \mathbf{~if~} c_2>0 \mathbf{~then~} c_2 \leftarrow c_2-1; \mathbf{goto~} n\!-\!1$ $\mathbf{else\ goto~} n$
	\item$n:\mathbf{halt}$
\end{itemize}
For $1<i<n-2$, the labels occurring in instruction $i$ belong to $\{0,\ldots,n-2\}$.
In a normalized two-counter program $\mathcal P$, given any initial values $v_1,v_2$,
 $\mathcal P$ halts with $v_1,v_2$ if and only if  $\mathcal P$ halts with initial values $0,0$.
 Moreover when $\mathcal P$ halts, the values of the counters are null.
The halting problem for normalized two-counter programs is also undecidable (see Lemme \ref{und-norm} in Appendix).

We now show that decisiveness is undecidable even for \emph{static} pCM, by considering only 
\emph{static} weights: for all $t\in \Delta$, $W(t)$ is a constant function. 


\begin{restatable}{theorem}{undecPCM}
\label{theorem:decisiveness-undecidable-static-pCM}
Decisiveness w.r.t. a finite set is undecidable in (static) pCM.
\end{restatable}

\section{Probabilistic safe one-counter machines}
\label{sec:oneCM}

We now study decisiveness for pCMs that only have  one counter denoted $c$. 
We also restrict $\Delta_1$:  a single counter PCM is \emph{safe} if for all $t\in \Delta_1$,
$(\mathbf{Pre}(t), \mathbf{Post}(t))\in \{1\}\times \{0,1,2\}$. In words, in a safe one-counter pCM,
a transition of $\Delta_1$ requires the counter to be positive and may either let it unchanged, 
or incremented or decremented by a unit. 

\subsection{One-state and one-counter pCM}

We first prove that decisiveness is undecidable for the probabilistic version of one-state and one-counter machines. 
Then we show how to restrict the weight functions and $\Delta_1$ such that this property becomes decidable.
Both proofs make use of the relationship between decisiveness and recurrence stated in Proposition~\ref{rec-dec-rec}, in an implicit way.
 

\begin{theorem}\label{dec-div-pPDA-undec}
The decisiveness problem for  safe one-counter pCM  is undecidable
even with a single state.
\end{theorem}
\begin{proof}
We will reduce the Hilbert's tenth problem to decisiveness problems.
Let $P\in \integer[X_1,\ldots X_k]$ be an integer polynomial with $k$ variables. This problem
asks whether there exist $n_1,\ldots, n_k\in \nat$ such that $P(n_1,\ldots,n_k)=0$.

\smallskip\noindent
We define $\mathcal C$ as follows. There are two transitions both in $\Delta_1$:
\begin{itemize}[nosep]
	\item $dec$ with $\mathbf{Pre}(dec)=1$ and $\mathbf{Post}(dec)=0$;
	\item $inc$ with $\mathbf{Pre}(inc)=0$ and $\mathbf{Post}(inc)=1$.
\end{itemize}
The weight of  $dec$ is the constant function 1, i.e., $W(dec, n)=f(n)=1$, while 
the weight of $inc$ is defined by the following (non polynomial) function:
\begin{small}
$$W(inc,  n)=g(n)=\min(P^2(n_1,\ldots,n_k)+1 \mid n_1+\ldots +n_k\leq n)$$
\end{small}

\noindent
This function is obviously computable. Let us study the decisiveness  of 
$\mathcal M_{\mathcal C}$ w.r.t. $s_0=1$ and $A=\{0\}$. Observe that
$\mathcal M_{\mathcal C}$ is the Markov chain $\mathcal M_1$ of Figure~\ref{fig:rw}.
Let us recall that in $\mathcal M_1$, the probability to reach $0$ from $i$
is 1 iff $\sum_{n\in \nat} \prod_{1\leq m<n} \rho_m =\infty$ and otherwise it is equal
to $\frac{\sum_{i\leq n} \prod_{1\leq m<n} \rho_m}{\sum_{n\in \nat} \prod_{1\leq m<n} \rho_m}$
with $\rho_m=\frac{1-p_m}{p_m}$.

\noindent
$\bullet$ Assume there exist  $n_1,\ldots, n_k\!\in\! \nat$ s.t. $P(n_1,\ldots,n_k)\!=\!0$. 
Let $n_0\!=\!n_1+\cdots + n_k$.
Thus for all $n\!\geq\! n_0$, $W(inc, n)=1$, which implies that
$\suiv_{\mathcal C}(n,n-1)=\suiv_{\mathcal C}(n,n+1)=\frac 1 2$.
Thus due to the results on $\mathcal M_1$, from any state $n$, one reaches $0$
almost surely and so $\mathcal M_{\mathcal C}$ is decisive.

\noindent
$\bullet$
Assume there do not exist  $n_1,\ldots, n_k\in \nat$ s.t. $P(n_1,\ldots,n_k)=0$.
For all $n\in \nat$, $W(inc, n)\geq 2$, implying that in $\mathcal M_1$,
$\rho_n\leq \frac{1}{2}$. Thus  $\mathcal M_{\mathcal C}$ is not decisive.

\end{proof}

Due to the negative result for single state and single counter pCM stated in 
Theorem~\ref{dec-div-pPDA-undec}, it is clear that one must restrict the possible weight functions.

\begin{restatable}{theorem}{decdivpda}
\label{dec-div-pPDA}
The decisiveness problem w.r.t. $s_0$ and finite $A$ for polynomial safe one-counter pCM  $\mathcal C$
with a single state
is decidable in linear time. 
\end{restatable}

\subsection{Homogeneous one-counter machines}


Let $\mathcal C$ be a one-counter safe pCM. For all $q\in Q$, 
let $S_{q,1}=\sum_{t= (q,\mathbf{Pre}(t),\mathbf{Post}(t),q^+_t)\in \Delta_1} W(t)$ and
$\mathbf{M}_{\mathcal C}$ be the $Q\times Q$ matrix
defined by:\\  
\centerline{$\mathbf{M}_{\mathcal C}[q,q']=\frac{\sum_{t= (q,\mathbf{Pre}(t),\mathbf{Post}(t),q')\in \Delta_1} W(t)}{S_{q,1}}$}
(thus  $\mathbf{M}_{\mathcal C}[q,q']$ is a function from $\nat$ to $\rat_\geq 0$).

\begin{definition}[pHM]
A \emph{probabilistic homogeneous machine} (pHM) is a probabilistic safe one-counter machine $\mathcal C= (Q,\Delta,W)$ where:
\begin{itemize}
%
%
%
	\item For all $t\in \Delta$, $W(t)$ is a positive polynomial in $\nat[X]$;
	\item For all $q,q'\in Q$, $\mathbf{M}_{\mathcal C}[q,q']$ is constant.
\end{itemize}
\end{definition}

Observe that by definition, in a pHM,  $\mathbf{M}_{\mathcal C}$ is a transition matrix.


%


\begin{example}
Here $\mathbf{M}_{\mathcal C}[q,q']=\mathbf{M}_{\mathcal C}[q,q'']=\frac{X^2+X+1}{2(X^2+X+1)}=\frac 1 2$.
\end{example}

 
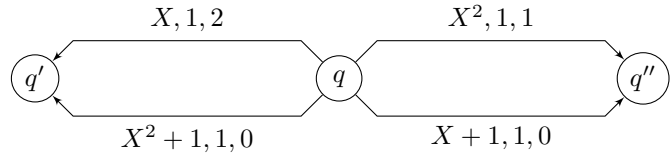
\begin{center}
  \begin{tikzpicture}[node distance=2cm,->,auto,-latex]
    
   \path (0,0) node[minimum size=0.6cm,draw,circle,inner sep=2pt] (qp) {$q'$};
    \path (4,0) node[minimum size=0.6cm,draw,circle,inner sep=2pt] (q) {$q$};
   \path (8,0) node[minimum size=0.6cm,draw,circle,inner sep=2pt] (qs) {$q''$};
  
  \draw[arrows=-latex'] (q)-- (3.5,0.5) --(0.5,0.5)node[pos=0.5,above] {$X,1,2$}--(qp) ;
 \draw[arrows=-latex'] (q)-- (3.5,-0.5) --(0.5,-0.5)node[pos=0.5,below] {$X^2+1,1,0$}--(qp) ;

 \draw[arrows=-latex'] (q)-- (4.5,-0.5) --(7.5,-0.5) node[pos=0.5,below] {$X+1,1,0$}--(qs) ;
   \draw[arrows=-latex'] (q)-- (4.5,0.5) -- (7.5,0.5)node[pos=0.5,above] {$X^2,1,1$}--(qs) ;

  \end{tikzpicture}
\end{center}

The family $(r_q)_{q\in Q}$ of the next proposition is independent of the function $W$ and is associated with the qualitative behaviour of $\mathcal C$.
\begin{restatable}{proposition}{qualpHM}
Let $\mathcal C$ be a pHM. 
Then one can compute in polynomial time 
a family $(r_q)_{q\in Q}$ such that for all $q$,
$r_q\in \{0,\ldots,|Q|-1\} \cup \{\infty\}$, and $Q\times \{0\}$ 
is reachable from $(q,k)$
iff $k\leq r_q$.
\end{restatable}

\begin{theorem}Let $\mathcal C$ be a pHM such that 
$\mathbf{M}_{\mathcal C}$ is irreducible. Then the decisiveness problem of $\mathcal C$
w.r.t. $s_0=(q,n)\in Q \times \nat$ 
and $A=Q\times \{0\}$ is decidable in polynomial time.
\end{theorem}
\begin{proof} 
With the notations of previous proposition,
assume that there exist $q$ with $r_q<\infty$ and $q'$ with $r_{q'}=\infty$. 
Since $\mathbf{M}_{\mathcal C}$ is irreducible, there is a sequence of transitions in $\Delta_1$
 $q_0\xrightarrow{1,v_1}q_1\cdots \xrightarrow{1,v_m} q_m$ with $q_0=q$
 and $q_m=q'$. Let $sv=\min(\sum_{i\leq j} (v_i-1)|j\leq m)$ and pick some $k>\max(r_q,-sv)$.
 Then there is a path in $\mathcal M_{\mathcal C}$ from $(q,k)$ to $(q',k+\sum_{i\leq m} (v_i-1))$, 
 which yields a contradiction since   $(q,k)$ cannot reach  $Q\times \{0\}$
 while $(q',k+\sum_{i\leq m} v_i)$ can reach it. Thus either (1) for all $q\in Q$,
 $r_q<\infty$ or (2) for all $q\in Q$, $r_q=\infty$. 

\smallskip\noindent
$\bullet$ First assume that for all $q \in Q$, 
$r_q<\infty$. Thus for all $k>r_q$, $(q,k)$ cannot reach  $Q\times \{0\}$
and thus  $\mathcal C$
is decisive w.r.t. $(q,k)$ and $Q\times \{0\}$. Now consider a configuration 
$(q,k)$ with $k\leq r_q$. By definition
there is a positive probability say $p_{(q,k)}$ to reach $Q\times \{0\}$ from $(q,k)$.
Let $p_{\min}=\min(p_{(q,k)}\mid q\in Q \wedge  k\leq r_q)$.
Then for all $(q,k)$ with $k\leq r_q$, there is a probability at least $p_{\min}$
to reach either  $Q\times \{0\}$ or $\{(q,k)\mid q\in Q \wedge  k> r_q\}$ by a path of length $\ell=\sum_{q\in Q} (r_q+1)$.
This implies that after $n\ell$ transitions the probability to reach either $Q\times \{0\}$ or $\{(q,k)\mid q\in Q \wedge  k> r_q\}$
is at least $1-(1-p_{\min})^n$.
Thus  $\mathcal C$
is decisive w.r.t. $(q,k)$ and $Q\times \{0\}$. Summarizing for all $(q,k)$, $\mathcal C$
is decisive w.r.t. $(q,k)$ and $Q\times \{0\}$.

\smallskip\noindent
$\bullet$ Now assume that for all $(q,k) \in Q\times \nat$, 
$Q\times \{0\}$ is reachable from $(q,k)$. Thus the decisiveness problem
boils down to the almost sure reachability of $Q\times \{0\}$.

\noindent
Since the target of decisiveness  is $Q\times \{0\}$, we can arbitrarily set up the outgoing 
transitions  of these states (i.e., $\Delta_0$) without changing the decisiveness problem. So we choose these  transitions 
and associated probabilities as follows.
For all $q,q'$ such that $\mathbf{M}_{\mathcal C}[q,q']>0$, there is a transition $t=q\xrightarrow{c,0}q'$
with $W(t)= \mathbf{M}_{\mathcal C}[q,q']$.

\noindent
Since $\mathbf{M}_{\mathcal C}$ is irreducible, there is a unique invariant distribution
$\pi_\infty$ (i.e., $\pi_\infty\mathbf{M}_{\mathcal C}=\pi_\infty$) 
fulfilling for all $q\in Q$, $\pi_\infty(q)>0$. 
 
\noindent
Let $(Q_n,N_n)_{n\in \nat}$ be the stochastic process defined by $\mathcal M_{\mathcal C}$
with $N_0=k$ for some $k$ and for all $q\in Q$, $\prob(Q_0=q)=\pi_\infty(q)$. Due to the invariance of $\pi_\infty$
and the choice of transitions for $Q\times \{0\}$, one gets by induction that  for all $n\in \nat$ :
\begin{itemize}
	\item $\prob(Q_n=q)=\pi_\infty(q)$;
	\item for all $k>0$, $\prob(N_{n+1}=k+v-1 | N_n=k)=$
	$\sum_{q\in Q}\pi_\infty(q)\frac{\sum_{t=(q,1,v,q')\in \Delta_1}W(t,k)}{S_{q,1}(k)}$\\
	$=\frac{\sum_{q\in Q}\pi_\infty(q)\prod_{q'\neq q}S_{q',1}(k)\sum_{t=(q,1,v,q')\in \Delta}W(t,k)}{\prod_{q'\in Q}S_{q',1}(k)}$;
	\item $\prob(N_{n+1}=0 | N_n=0)=1$.
\end{itemize}
 
\noindent
For $v\in \{-1,0,1\}$, let us define the polynomial $P_{v}$ by:
$$\sum_{q\in Q}\pi_\infty(q)\prod_{q'\neq q}S_{q',1}\sum_{t=(q,1,v+1,q')\in \Delta_1}W(t)$$
Due to the previous observations, the stochastic process $(N_n)_{n\in \nat}$
is the Markov chain defined below where the weights outgoing from a state
have to be normalized:
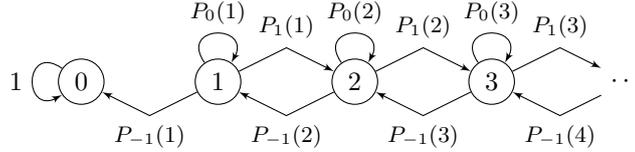
\begin{center}
  \begin{tikzpicture}[node distance=2cm,->,auto,-latex,scale=0.9]
    
   \path (0,0) node[minimum size=0.6cm,draw,circle,inner sep=2pt] (q0) {$0$};
   \path (2,0) node[minimum size=0.6cm,draw,circle,inner sep=2pt] (q1) {$1$};
   \path (4,0) node[minimum size=0.6cm,draw,circle,inner sep=2pt] (q2) {$2$};
   \path (6,0) node[minimum size=0.6cm,draw,circle,inner sep=2pt] (q3) {$3$};
   \path (8,0) node[] (q4) {$\cdots$};

      

    \draw[arrows=-latex'] (q1)-- (1,-0.5) node[pos=1,below] {{\footnotesize $P_{-1}(1)$}}--(q0) ;

    \draw[arrows=-latex'] (q1)-- (3,0.5) node[pos=1,above] {{\footnotesize $P_{1}(1)$}}--(q2) ;
    \draw[arrows=-latex'] (q2)-- (3,-0.5) node[pos=1,below] {{\footnotesize$P_{-1}(2)$}}--(q1) ;

    \draw[arrows=-latex'] (q2)-- (5,0.5) node[pos=1,above] {{\footnotesize$P_{1}(2)$}}--(q3) ;
    \draw[arrows=-latex'] (q3)-- (5,-0.5) node[pos=1,below] {{\footnotesize$P_{-1}(3)$}}--(q2) ;

    \draw[arrows=-latex'] (q3)-- (7,0.5) node[pos=1,above] {{\footnotesize$P_{1}(3)$}}--(q4) ;
    \draw[arrows=-latex'] (q4)-- (7,-0.5) node[pos=1,below] {{\footnotesize$P_{-1}(4)$}}--(q3) ;

   \draw[-latex'] (q0) .. controls +(145:30pt) and +(215:30pt) .. (q0) node[pos=.5,left] {$1$};
   \draw[-latex'] (q1) .. controls +(125:30pt) and +(55:30pt) .. (q1) node[pos=.5,above] {{\footnotesize $P_{0}(1)$}};
   \draw[-latex'] (q2) .. controls +(125:30pt) and +(55:30pt) .. (q2) node[pos=.5,above] {{\footnotesize $P_{0}(2)$}};
   \draw[-latex'] (q3) .. controls +(125:30pt) and +(55:30pt) .. (q3) node[pos=.5,above] {{\footnotesize $P_{0}(3)$}};

  \end{tikzpicture}
\end{center}
Using our hypothesis about reachability, $P_{-1}$ is a positive polynomial (while $P_1$ could be null)
and thus  the decisiveness of this Markov chain w.r.t. state $0$ is equivalent to 
the decisiveness of the Markov chain below:
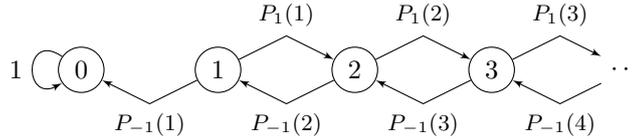
\begin{center}
  \begin{tikzpicture}[node distance=2cm,->,auto,-latex,scale=0.9]
    
   \path (0,0) node[minimum size=0.6cm,draw,circle,inner sep=2pt] (q0) {$0$};
   \path (2,0) node[minimum size=0.6cm,draw,circle,inner sep=2pt] (q1) {$1$};
   \path (4,0) node[minimum size=0.6cm,draw,circle,inner sep=2pt] (q2) {$2$};
   \path (6,0) node[minimum size=0.6cm,draw,circle,inner sep=2pt] (q3) {$3$};
   \path (8,0) node[] (q4) {$\cdots$};

      

    \draw[arrows=-latex'] (q1)-- (1,-0.5) node[pos=1,below] {{\footnotesize $P_{-1}(1)$}}--(q0) ;

    \draw[arrows=-latex'] (q1)-- (3,0.5) node[pos=1,above] {{\footnotesize $P_{1}(1)$}}--(q2) ;
    \draw[arrows=-latex'] (q2)-- (3,-0.5) node[pos=1,below] {{\footnotesize$P_{-1}(2)$}}--(q1) ;

    \draw[arrows=-latex'] (q2)-- (5,0.5) node[pos=1,above] {{\footnotesize$P_{1}(2)$}}--(q3) ;
    \draw[arrows=-latex'] (q3)-- (5,-0.5) node[pos=1,below] {{\footnotesize$P_{-1}(3)$}}--(q2) ;

    \draw[arrows=-latex'] (q3)-- (7,0.5) node[pos=1,above] {{\footnotesize$P_{1}(3)$}}--(q4) ;
    \draw[arrows=-latex'] (q4)-- (7,-0.5) node[pos=1,below] {{\footnotesize$P_{-1}(4)$}}--(q3) ;

   \draw[-latex'] (q0) .. controls +(145:30pt) and +(215:30pt) .. (q0) node[pos=.5,left] {$1$};

  \end{tikzpicture}
\end{center}

\noindent
Due to Theorem~\ref{dec-div-pPDA}, this problem is decidable (in linear time)
and either (1)  for all $k \in \nat$ this Markov chain is decisive w.r.t $k$ and $0$
or (2) for all $k>0$ this Markov chain is not decisive w.r.t $k$ and $0$.
Let us analyze the two cases w.r.t. the Markov chain of the pHM.

\noindent
{\bf Case (1)} In the stochastic process $(Q_n,N_n)_{n\in \nat}$,
the initial distribution has a positive probability for $(q,k)$ for all $q\in Q$.
This implies that  for all $q$, $\mathcal C$ is decisive w.r.t. $(q,k)$ and  $Q\times \{0\}$.
Since $k$ was arbitrary, this means that for all $(q,k)$, $\mathcal C$ is decisive w.r.t. $(q,k)$ and  $Q\times \{0\}$.

\noindent
{\bf Case (2)} Choosing $k=1$ and applying the same reasoning as for the previous case,
there is some $(q,1)$ which is not decisive (and so for all $(q,k')$ with $k'>0$). Let $q'\in Q$, since  $\mathbf{M}_{\mathcal C}$
is irreducible, there is a (shortest) sequence of transitions in $\Delta_1$ leading from $q'$ to $q$ whose length
is at most $|Q|-1$. Thus for all $(q',k')$ with $k'\geq |Q|$ there is a positive probability to reach some
$(q,k)$ with $k>0$. Thus $(q',k)$ is not decisive. 

\noindent
Now let $(q',k')$ with $k'<|Q|$. Then we compute by a breadth first exploration
the configurations reachable from $(q',k')$ until either (1) one reaches some $(q'',k'')$
with $k''\geq |Q|$ or (2) the full (finite) reachability set is computed. In the first case,
there is a positive probability to reach some $(q'',k'')$ with $k''\geq |Q|$ and from $(q'',k'')$ to some
$(q,k)$ with $k>0$ and so $(q',k')$ is not decisive. In the second case, 
it means that the reachable set is finite and from any configuration of this
set there is a positive probability to reach $Q\times \{0\}$ by a path of length at most the size of this
set. Thus almost surely $Q\times \{0\}$ will be reached and $(q',k')$ is decisive.

\end{proof}

\section{Probabilistic Petri nets}
\label{sec:pPN}

\smallskip
We now introduce probabilistic Petri nets as a subclass of pCM.

\begin{definition}[pPN] A \emph{probabilistic Petri net (pPN)} $\mathcal N$
is a pCM $\mathcal N = (Q,P,\Delta,W)$ where $Q$ is a singleton and $\Delta_0=\emptyset$. 
\end{definition}

\noindent
{\bf Notations.} Since there is a unique control state in a pPN, a configuration in a pPN is reduced to $\mathbf{m} \in \nat^P$ and it is called a \emph{marking}.
A pair $(\mathcal N,\mathbf{m}_0)$, where $\mathcal N$ is a pPN and $\mathbf{m}_0\in  \nat^P$
is some (initial) marking, is called a \emph{marked pPN}.

In previous works~\cite{AbdullaHM07,DBLP:conf/lics/BrazdilCK0VZ18} about pPNs, 
the weight function $W$ is a \emph{static} one: i.e., a function from $\Delta$ to $\nat^*$.
As above, we call these models \emph{static} probabilistic Petri nets.
\begin{figure}[!htb]
\begin{center}
\minipage{0.39\textwidth}
\begin{tikzpicture}[xscale=0.55,yscale=0.55]

\path (0,0) node[draw,circle,inner sep=2pt,minimum size=0.6cm,
label={[xshift=0cm, yshift=0cm]$p_i$}] (pi) {};

\path (0,-2) node[draw,circle,inner sep=2pt,minimum size=0.6cm,
label={[xshift=0cm, yshift=0cm]$c_j$}] (cj) {};

\path (0,-4) node[draw,circle,inner sep=2pt,minimum size=0.6cm,
label={[xshift=0cm, yshift=0cm]$p_{i'}$}] (piprime) {};

\path (2,-4) node[draw,circle,inner sep=2pt,minimum size=0.6cm,
label={[xshift=0.8cm, yshift=-0.6cm]$stop$}] (stop) {};

\path (-2,-2) node[draw,rectangle,inner sep=2pt,minimum width=0.4cm,minimum height=0.2cm,
label={[xshift=-0.8cm, yshift=-0.4cm]$inc_i$}] (inci) {};

\path (2,-2) node[draw,rectangle,inner sep=2pt,minimum width=0.4cm,minimum height=0.2cm,
label={[xshift=0.8cm, yshift=-0.4cm]$exit_i$}] (exiti) {};

\draw[arrows=-latex] (pi) -- (-2,0) --(inci) ;
\draw[arrows=-latex] (pi) -- (2,0) --(exiti) ;

\draw[arrows=-latex] (inci)--(cj);

\draw[arrows=-latex] (inci)--(-2,-4)--(piprime);
\draw[arrows=-latex] (exiti)  --(stop) ;

\end{tikzpicture}
\caption{$i: c_j \leftarrow c_j+1; \mathbf{ goto~} i'$}
\label{fig:increment}


\endminipage\hfill
\minipage{0.60\textwidth}

\begin{tikzpicture}[xscale=0.5,yscale=0.5]

\path (0,0) node[draw,circle,inner sep=2pt,minimum size=0.6cm,
label={[xshift=0cm, yshift=0cm]$p_n$}] (pn) {};

\path (0,-2) node[draw,circle,inner sep=2pt,minimum size=0.6cm,
label={[xshift=0cm, yshift=0cm]$sim$}] (sim) {};

\path (0,-4) node[draw,circle,inner sep=2pt,minimum size=0.6cm,
label={[xshift=0cm, yshift=-1.2cm]$p_0$}] (p0) {};

\path (4,-4) node[draw,circle,inner sep=2pt,minimum size=0.6cm,
label={[xshift=0cm, yshift=-1.2cm]$stop$}] (stop) {};

\path (2,0) node[draw,circle,inner sep=2pt,minimum size=0.6cm,
label={[xshift=0cm, yshift=0cm]$c_1$}] (c1) {};

\path (6,0) node[draw,circle,inner sep=2pt,minimum size=0.6cm,
label={[xshift=0cm, yshift=0cm]$c_2$}] (c2) {};

\path (-2,-2) node[draw,rectangle,inner sep=2pt,minimum width=0.4cm,minimum height=0.2cm,
label={[xshift=-0.8cm, yshift=-0.4cm]$again$}] (again) {};

\path (2,-2) node[draw,rectangle,inner sep=2pt,minimum width=0.4cm,minimum height=0.2cm,
label={[xshift=0.8cm, yshift=-0.4cm]$clean_1$}] (clean1) {};

\path (6,-2) node[draw,rectangle,inner sep=2pt,minimum width=0.4cm,minimum height=0.2cm,
label={[xshift=0.8cm, yshift=-0.4cm]$clean_2$}] (clean2) {};

\path (2,-4) node[draw,rectangle,inner sep=2pt,minimum width=0.4cm,minimum height=0.2cm,
label={[xshift=0cm, yshift=-1cm]$clean_3$}] (clean3) {};

\draw[arrows=-latex] (pn) -- (-2,0) --(again) ;
\draw[arrows=-latex] (again)--(sim);
\draw[arrows=-latex] (again)--(-2,-4)--(p0);

\draw[arrows=-latex] (clean1)  --(stop) ;
\draw[arrows=-latex] (stop)  --(clean1) ;
\draw[arrows=-latex] (c1)  --(clean1) ;

\draw[arrows=-latex] (clean2)  --(stop) ;
\draw[arrows=-latex] (stop)  --(clean2) ;
\draw[arrows=-latex] (c2)  --(clean2) ;

\draw[arrows=-latex] (clean3)  --(stop) ;
\draw[arrows=-latex] (stop)  --(clean3) ;
\draw[arrows=-latex] (sim)  --(clean3) ;

\end{tikzpicture}
\caption{halt instruction and cleaning stage}
\label{fig:halt}

\endminipage
\end{center}
\end{figure}

\begin{figure}
\begin{center}


\begin{tikzpicture}[xscale=0.6,yscale=0.6]

\path (0,0) node[draw,circle,inner sep=2pt,minimum size=0.6cm,
label={[xshift=0cm, yshift=0cm]$p_i$}] (pi) {};

\path (0,-2) node[draw,circle,inner sep=2pt,minimum size=0.6cm,
label={[xshift=0cm, yshift=0cm]$c_j$}] (cj) {};

\path (2,-2) node[draw,circle,inner sep=2pt,minimum size=0.6cm,
label={[xshift=0.6cm, yshift=-0.6cm]$q_i$}] (qi) {};

\path (-3,-4) node[draw,circle,inner sep=2pt,minimum size=0.6cm,
label={[xshift=0cm, yshift=-1.2cm]$p_{i'}$}] (piprime) {};

\path (2,-4) node[draw,circle,inner sep=2pt,minimum size=0.6cm,
label={[xshift=0cm, yshift=-1.2cm]$p_{i''}$}] (pisecond) {};

\path (5,-4) node[draw,circle,inner sep=2pt,minimum size=0.6cm,
label={[xshift=0cm, yshift=-1.2cm]$stop$}] (stop) {};

\path (0,-5.5) node[draw,circle,inner sep=2pt,minimum size=0.6cm,
label={[xshift=-0.8cm, yshift=-0.5cm]$\ \ sim$}] (sim) {};

\path (-3,-2) node[draw,rectangle,inner sep=2pt,minimum width=0.4cm,minimum height=0.2cm,
label={[xshift=-0.6cm, yshift=-0.4cm]$dec_i$}] (deci) {};

\path (2,-1) node[draw,rectangle,inner sep=2pt,minimum width=0.4cm,minimum height=0.2cm,
label={[xshift=0.8cm, yshift=-0.4cm]$begZ_i$}] (begZi) {};

\path (2,-3) node[draw,rectangle,inner sep=2pt,minimum width=0.4cm,minimum height=0.2cm,
label={[xshift=0.8cm, yshift=-0.4cm]$endZ_i$}] (endZi) {};

\path (5,-2) node[draw,rectangle,inner sep=2pt,minimum width=0.4cm,minimum height=0.2cm,
label={[xshift=0.8cm, yshift=-0.4cm]$exit_i$}] (exiti) {};

\path (0,-4) node[draw,rectangle,inner sep=2pt,minimum width=0.4cm,minimum height=0.2cm,
label={[xshift=-0.6cm, yshift=-0.4cm]$rm_i$}] (rmi) {};

\draw[arrows=-latex] (pi) -- (-3,0) --(deci) ;
\draw[arrows=-latex] (pi) -- (2,0) --(begZi) ;
\draw[arrows=-latex] (begZi) --(qi) ;
\draw[arrows=-latex] (qi) -- (endZi);
\draw[arrows=-latex] (cj)--(deci);
\draw[arrows=-latex] (endZi)--(pisecond);
\draw[arrows=-latex] (deci)--(piprime);

\draw[arrows=-latex] (sim) --(rmi) node[pos=0.5,right] {$2$}  ;
\draw[arrows=-latex] (exiti)  --(stop) ;

\draw[arrows=-latex] (pi)  --(5,0)--(exiti) ;
\draw[arrows=-latex] (qi)  --(rmi) ;
\draw[arrows=-latex] (rmi)  --(qi) ;
\draw[arrows=-latex] (cj)  --(rmi) ;
\draw[arrows=-latex] (rmi)  --(cj) ;

\end{tikzpicture}
\end{center}
\caption{\begin{small}$i:\mathbf{if~} c_j>0 \mathbf{~then~} c_j \leftarrow c_j-1; \mathbf{goto~} i' \mathbf{else\ goto~} i''$\end{small}}
\label{fig:decrement}

\end{figure}

Static-probabilistic VASS (and so pPNs) are \emph{decisive} with respect to \emph{upward closed sets} (Corollary 4.4 in \cite{AbdullaHM07}) but they may not be decisive w.r.t. an arbitrary finite set. 
Surprisingly, the decisiveness problem for Petri nets or VASS seems not to have been studied. We establish below that even for polynomial pPNs, decisiveness is undecidable.



\begin{theorem}\label{theorem:polynomial-pPN-decisiveness-undecidable}
The decisiveness problem of polynomial pPNs w.r.t. a finite or upward closed set is undecidable.
%
%
\end{theorem}

\begin{proof}
We reduce the reachability problem of normalized two-counter machines to the  decisiveness problem of pPN.
Let $\mathcal C$ be a normalized two-counter machine with an instruction set $\{0,\ldots,n\}$. The  corresponding marked pPN $(\mathcal N_\mathcal C, \mathbf{m}_0)$ is built as follows.
Its set of places is $P=\{p_i \mid 0\leq i\leq n\} \cup \{q_i \mid i \mbox{~is a test instruction}\}   \cup \{c_j \mid 1\leq j\leq 2\} \cup \{sim,stop\}$.
The initial marking is $ \mathbf{m}_0 = p_0$.

\smallskip\noindent
The set $\Delta$ of transitions is defined by a pattern per type of instruction. The pattern for the incrementation  instruction
is depicted in Figure~\ref{fig:increment}.  The pattern for the test  instruction 
is depicted in Figure~\ref{fig:decrement}.  The pattern for the halt  instruction
is depicted in Figure~\ref{fig:halt} with in addition a \emph{cleaning stage}.

\smallskip\noindent
Before specifying the weight function $W$, let us describe the qualitative behaviour of this net.
$(\mathcal N_\mathcal C,\mathbf{m}_0)$ performs repeatedly a \emph{weak} simulation of  $\mathcal C$. As usual
 since the zero test does not exist in Petri nets, during a test instruction $i$, the simulation can follow the zero branch while the corresponding counter
is non null (transitions  $begZ_i$ and $endZ_i$). If the net has cheated then with transition $rm_i$, it can remove
tokens from $sim$ (two per two). In addition when the instruction is not $\mathbf{halt}$, instead of simulating it, it can \emph{exit}
the simulation by putting a token in $stop$ and then will remove tokens from the counter places including the simulation counter
as long as they are not empty.
The simulation of the $\mathbf{halt}$ instruction consists in restarting the simulation and incrementing the simulation counter
$sim$.

\smallskip\noindent
Thus the set of reachable markings is included in the following set of markings 
$\{p_i +xc_1+yc_2+z sim\mid 0\leq i\leq n, x,y,z \in \nat\} \cup \{q_i +xc_1+yc_2+z sim\mid  i \mbox{~is a test instruction}, x,y,z \in \nat\} \cup \{stop +xc_1+yc_2+z sim\mid x,y,z \in \nat\}$.
By construction, the marking $stop$ is always reachable. We will establish that $\mathcal N_\mathcal C$ is decisive w.r.t. $\mathbf{m}_0$ and $\{stop\}$ if and only if $\mathcal C$ does not halt.

\smallskip\noindent
Let us specify the weight function. For any incrementation instruction $i$, 
$W(inc_i,\mathbf{m})=\mathbf{m}(sim)^2+1$.
For any test instruction $i$, 
$W(begZ_i,\mathbf{m})=\mathbf{m}(sim)^2+1$,
$W(dec_i,\mathbf{m})=2\mathbf{m}(sim)^4+2$
and $W(rm_i,\mathbf{m})=2$.
All other weights are equal to 1.

\smallskip\noindent
$\bullet$ Assume that $\mathcal C$ halts and consider its execution $\sigma_\mathcal C$ with initial values $(0,0)$.
Let $\ell=|\sigma_\mathcal C|$ be the length of this execution. Consider now $\sigma$ the infinite sequence of $(\mathcal N_\mathcal C,\mathbf{m}_0)$ 
that infinitely performs the correct simulation of this execution. The infinite sequence $\sigma$ never marks the place $stop$.
We now show that the probability of $\sigma$ is non null implying that $\mathcal N_\mathcal C$ is not
decisive.

\smallskip\noindent
After every simulation of $\sigma_\mathcal C$, the marking of $sim$ is incremented and it is never decremented
since (due to the correctness of the simulation) every time a transition $begZ_i$ is fired, the corresponding counter place
$c_j$ is unmarked which forbids the firing of $rm_i$. So during the $(n+1)^{th}$ simulation of $\rho$,
the marking of $sim$ is equal to $n$.

\smallskip\noindent
So consider the  probability of the correct simulation of an instruction $i$  during the $(n+1)^{th}$ simulation.
\begin{itemize}
	\item If $i$ is an incrementation then the weight of $inc_i$ is $n^2$ and the weight of $exit_i$ is 1.
	So the probability of a correct simulation is 
	$\frac{n^2+1}{n^2+2}= 1-\frac{1}{n^2+2}\geq e^{-\frac{2}{n^2+2}}$.~\footnote{We use $1-x\geq e^{-2x}$ for $0\leq x\leq \frac 1 2$.}
	\item If $i$ is a test of $c_j$ and the marking of $c_j$ is non null 
	then the weight of $dec_i$ is $2n^4+2$, the weight of $begZ_i$ is $n^2+1$ and the weight of $exit_i$ is 1.
	So the probability of a correct simulation is 
	$\frac{2n^4+2}{2n^4+n^2+4}\geq \frac{2n^4+2}{2n^4+2n^2+4}=  \frac{n^2+1}{n^{2}+2}= 1-\frac{1}{n^2+2}\geq e^{-\frac{2}{n^2+2}}$.
	\item If $i$ is a test of $c_j$ and the marking of $c_j$ is null 
	then the weight of $begZ_i$ is $n^2+1$ and the weight of $exit_i$ is 1.
	So the probability of a correct simulation is 
	$\frac{n^2+1}{n^2+2}= 1-\frac{1}{n^2+2}\geq e^{-\frac{2}{n^2+2}}$.
\end{itemize}
So the probability of the correct simulation during the $(n+1)^{th}$ simulation is at least $(e^{-\frac{2}{n^2+2}})^\ell=e^{-\frac{2\ell}{n^2+2}}$.
Hence the probability of $\sigma$ is at least $\prod_{n\in \nat}e^{-\frac{2\ell}{n^2+2}}=e^{-\sum_{n\in \nat} \frac{2\ell}{n^2+2}}>0$, 
as the sum in the exponent converges.

\smallskip\noindent
$\bullet$ Assume that $\mathcal C$ does not halt (and so does not halt for any initial values of the counters).
We partition the set of infinite paths into a countable family of subsets and prove that for all of them the probability
to infinitely avoid to mark $stop$ is null which will imply that $\mathcal N_{\mathcal C}$ is decisive. 
The partition is based on $k \in \nat \cup \{\infty\}$, the number of firings of $again$ in the path.

\smallskip\noindent
{\bf Case $k<\infty$}. Let $\sigma$ be such a path and consider the suffix of $\sigma$
after the last firing of $again$. The marking of $sim$ is at most $k$ and can only decrease along the suffix.
Consider a simulation of an increment instruction $i$. The weight of $inc_i$ is at most 
 is $k^2+1$ and the weight of $exit_i$ is 1.
	So the probability of avoiding   $exit_i$ is at most
	$\frac{k^2+1}{k^2+2}= 1-\frac{1}{k^2+2}\leq e^{-\frac{1}{k^2+2}}$.
        Consider the simulation of a test instruction $i$.
	Then the weight of $dec_i$ is at most $2k^{4}+2$, the weight of $begZ_i$ is at most
	$k^2+1$ and the weight of $exit_i$ is 1.
	So the probability  of avoiding   $exit_i$ is at most
	$\frac{2k^{4}+k^2+2}{2k^{4}+k^2+4}\leq \frac{4k^{4}+1}{4k^{4}+2}=1-\frac{1}{4k^{4}+2}\leq e^{-\frac{1}{4k^{4}+2}}$.
	
\smallskip\noindent
Thus after $n$ simulations of instructions in the suffix, the probability to avoid to mark $stop$
is at most 	$e^{-\frac{n}{4k^{4}+2}}$. Letting $n$ go to infinity yields the result.

\smallskip\noindent
{\bf Case $k=\infty$}. We first show that almost surely there will be an infinite number of
simulations of $\mathcal C$ with the marking of $sim$ at most 1. Observe that all these simulations
are incorrect since they mark $p_n$ while $\mathcal C$ does not halt. So at least once per simulation
some place $q_i$ and the corresponding counter $c_j$ must be marked and if the marking of $sim$ is at least 2 with probability $\frac 2 3$ two tokens
of $sim$ are removed (recall that the weight of $rm_i$ is 2 and the weight of $endZ_i$ is 1). Thus once the marking of $sim$ is greater than 1,
considering the successive random markings of $sim$ after the firing of $again$ until it possibly reaches 1, this behaviour is \emph{stochastically bounded} by the following random
walk:

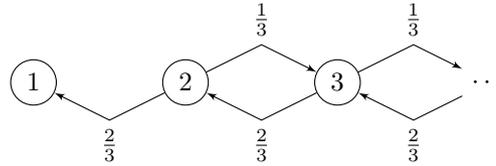
\vspace{-1mm}
\begin{center}
  \begin{tikzpicture}[node distance=2cm,->,auto,-latex]
    
   \path (2,0) node[minimum size=0.6cm,draw,circle,inner sep=2pt] (q1) {$1$};
   \path (4,0) node[minimum size=0.6cm,draw,circle,inner sep=2pt] (q2) {$2$};
   \path (6,0) node[minimum size=0.6cm,draw,circle,inner sep=2pt] (q3) {$3$};
   \path (8,0) node[] (q4) {$\cdots$};

   \draw[arrows=-latex'] (q2)-- (3,-0.5) node[pos=1,below] {$\frac 2 3$}--(q1) ;

    \draw[arrows=-latex'] (q2)-- (5,0.5) node[pos=1,above] {$\frac 1 3$}--(q3) ;
    \draw[arrows=-latex'] (q3)-- (5,-0.5) node[pos=1,below] {$\frac 2 3$}--(q2) ;

    \draw[arrows=-latex'] (q3)-- (7,0.5) node[pos=1,above] {$\frac 1 3$}--(q4) ;
    \draw[arrows=-latex'] (q4)-- (7,-0.5) node[pos=1,below] {$\frac 2 3$}--(q3) ;

  \end{tikzpicture}
\end{center}
\vspace{-2mm}

In this random walk, one reaches the state 1 with probability 1. This establishes
that almost surely there will be an infinite number of simulations of $\mathcal C$ with the marking of $sim$
at most 1. Such a simulation must simulate at least one instruction. If this instruction is an incrementation,
the exiting probability is at least $\frac {1}{3}$; if it is a test instruction the exiting probability is at least $\frac {1}{7}$.
Thus after $n$ such simulations of  $\mathcal C$, the probability to avoid to mark $stop$ is at most $(\frac {6} {7})^n$.
Letting $n$ go to infinity yields the result.

\noindent
Observe that the result remains true when substituting the singleton $\{stop\}$ by the set of markings
greater than or equal to $stop$.
\end{proof}

\vspace{2mm}
We deduce thus that decisiveness of extended (probabilistic) Petri nets is undecidable : 
in particular for  Reset Petri nets \cite{DBLP:conf/icalp/DufourdFS98}, Post-Self-Modifying Petri nets \cite{DBLP:conf/icalp/Valk78}, Recursive Petri nets, etc.

\begin{definition}
The \emph{language} of a marked Petri net $(\mathcal N, \mathbf{m}_0)$
is defined by $\mathcal L(\mathcal N,  \mathbf{m}_0)=\{\sigma\in \Delta^* \mid \mathbf{m}_0 \xrightarrow{\sigma} \}$.
$(\mathcal N \mathbf{m}_0)$ is \emph{regular} if $\mathcal L(\mathcal N,  \mathbf{m}_0)$ is regular.
\end{definition}

Given a marked Petri net $(\mathcal N, \mathbf{m}_0)$, the problem who asks whether it is regular is decidable \cite{GinzburgY80,DBLP:journals/jcss/ValkV81}
and belongs to {\sf EXPSPACE}~\cite{Demri13}. For establishing the next proposition, we only need the following result:
There exists a computable bound $B(\mathcal N,\mathbf{m}_0)$ such that for all markings $\mathbf{m}_1$ reachable from
$\mathbf{m}_0$ and all markings $\mathbf{m}_2$ with some
$p \in P$ fulfilling $\mathbf{m}_2(p)+B(\mathcal N,\mathbf{m}_0)<\mathbf{m}_1(p)$, $\mathbf{m}_2$ is unreachable from $\mathbf{m}_1$~(\cite{GinzburgY80}).

\begin{theorem}
\label{theorem:decidable:regularPN}
Let $(\mathcal N,\mathbf{m}_0)$ be a regular marked pPN and $\mathbf{m}_1$ be a marking. Then  $(\mathcal N,\mathbf{m}_0)$ is decisive with respect to 
$\mathbf{m}_0$ and $\{\mathbf{m}_1\}$.
\end{theorem}
\begin{proof}
Consider the following  algorithm that, after computing $B(\mathcal N,\mathbf{m}_0)$,  builds the following  finite graph:
\begin{itemize}
	\item Push on the stack $\mathbf{m}_0$.
	\item While the stack is not empty, pop from the stack some marking $\mathbf{m}$.
	Compute the set of transition firings $\mathbf{m}\xrightarrow{t}\mathbf{m}'$.
	Push on the stack $\mathbf{m}'$ if:
	\begin{enumerate}
		\item $\mathbf{m}'$ is not already present in the graph,
		\item and $\mathbf{m}'\neq \mathbf{m}_1$,
		\item and for all $p\in P$, $\mathbf{m}_1(p)+B(\mathcal N,\mathbf{m}_0)\geq \mathbf{m}'(p)$.
	\end{enumerate}
\end{itemize}
Due to the third condition, this algorithm terminates. From above, if $\mathbf{m}_1$
does not occur in the graph then $\mathbf{m}_1$ is unreachable from $\mathbf{m}_0$ and thus $\mathcal N$
is decisive w.r.t. $\mathbf{m}_1$.

\noindent
Otherwise, considering the weights specified by $W$ and adding loops for states without successors, this graph can be viewed as a finite Markov chain
and so reaching some bottom strongly connected component (BSCC) almost surely. There are three possible cases: (1) the BSCC consisting of 
$\mathbf{m}_1$, (2) a BSCC consisting of a single marking  $\mathbf{m}$ for which there exists some $p\in P$ fulfilling $\mathbf{m}_1(p)+B(\mathcal N,\mathbf{m}_0)< \mathbf{m}(p)$
and thus from which $\mathbf{m}_1$ is unreachable or (3) a BSCC that is also a BSCC of $\mathcal M_\mathcal N$ and thus from which one cannot reach
$\mathbf{m}_1$. This establishes that $\mathcal N$ is decisive w.r.t. $\mathbf{m}_1$.

\end{proof}

\noindent
{\bf Observation.} In this particular case, instead of using Algorithm~\ref{algo:prob-reach-dec} to frame the reachability probability, one can use the Markov chain of the proof
to exactly compute this probability.

\section{Conclusion and perspectives}
\label{sec:conclusion}

We have studied the decidability of decisiveness with respect to several subclasses of probabilistic counter machines.
The results are summarized in the following table.
When $A$ is not mentioned it means that $A$ is finite.

\medskip
\begin{center}
\noindent
\begin{scriptsize}
\begin{tabular}{|l|l|l|l|}
  \hline
  model & \!\!constant &\!\!polynomial &\!\!general\\
  \hline
pHM & D & D [Th \ref{dec-div-pPDA}] & U [Th \ref{dec-div-pPDA-undec}] \\
&  &  &  even with a single state\\
  \hline
  pPN & ? & U [Th \ref{theorem:polynomial-pPN-decisiveness-undecidable}] & U \\
& & also w.r.t. upward closed sets[Th \ref{theorem:polynomial-pPN-decisiveness-undecidable}] & but D when regular [Th \ref{theorem:decidable:regularPN}] \\
  \hline
  pCM & U [Th \ref{theorem:decisiveness-undecidable-static-pCM}] & U    & U    \\
  \hline
\end{tabular}
\end{scriptsize}
\end{center}
In the future, apart for solving the left open problem in the above table,
we plan to introduce sufficient conditions for decisiveness for models
with undecidability of decisiveness like pPNs with polynomial weights. This could have 
a practical impact for real case-study modellings. 
 
 In another direction, we have established that the decisiveness and recurrence properties are closely related.
 It would be interesting to define a property related to transience in Markov chains. In fact we have identified
 such a property called divergence and the definition and analysis of this property will appear in a forthcoming paper.


\newpage
\bibliography{ref}
\section*{Appendix}
\label{sec:apprendix-serge}

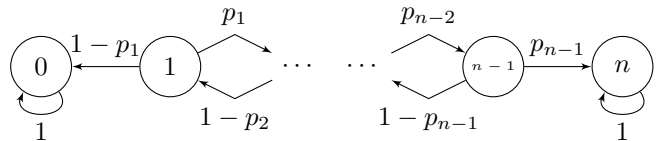
\begin{figure}[b]
\begin{center}
  \begin{tikzpicture}[node distance=2cm,->,auto,-latex,scale=0.85]
    
   \path (0,0) node[minimum size=0.8cm,draw,circle,inner sep=2pt] (q0) {$0$};
   \path (2,0) node[minimum size=0.8cm,draw,circle,inner sep=2pt] (q1) {$1$};
   \path (4,0) node[] (q2) {$\cdots$};
   \path (5,0) node[] (q3) {$\cdots$};
   \path (7,0) node[minimum size=0.8cm,draw,circle,inner sep=2pt] (q4) {\tiny{$n-1$}};
   \path (9,0) node[minimum size=0.8cm,draw,circle,inner sep=2pt] (q5) {$n$};

      

    \draw[arrows=-latex'] (q1) --(q0) node[pos=0.5,above] {$1-p_1$};

    \draw[arrows=-latex'] (q1)-- (3,0.5) node[pos=1,above] {$p_1$}--(q2) ;
    \draw[arrows=-latex'] (q2)-- (3,-0.5) node[pos=1,below] {$1-p_2$}--(q1) ;


    \draw[arrows=-latex'] (q3)-- (6,0.5) node[pos=1,above] {$p_{n-2}$}--(q4) ;
    \draw[arrows=-latex'] (q4)-- (6,-0.5) node[pos=1,below] {$1-p_{n-1}$}--(q3) ;

    \draw[arrows=-latex'] (q4) --(q5) node[pos=0.5,above] {$p_{n-1}$};

   \draw[-latex'] (q0) .. controls +(-55:30pt) and +(-125:30pt) .. (q0) node[pos=.5,below] {$1$};
   \draw[-latex'] (q5) .. controls +(-55:30pt) and +(-125:30pt) .. (q5) node[pos=.5,below] {$1$};

  \end{tikzpicture}
\end{center}
\caption{The (generalized) gambler's ruin}
\label{fig:gambler}
\end{figure}

\begin{proposition}
\label{prop:gambler}
 For all  $0\leq m\leq n$, let $preach_m$ be the probability to reach state $0$
from state $m$ in the Markov chain of Figure~\ref{fig:gambler}. Then:
$$preach_{m}=\frac{\sum_{m \leq j< n}\prod_{0<k\leq j}\rho_k}{\sum_{0\leq j< n}\prod_{0<k\leq j}\rho_k}$$
with for all $m$, $\rho_m=\frac{1-p_m}{p_m}$.
\end{proposition}
\begin{proof}
$preach_0=1$, $preach_n=0$ and for $0<m<n$,
$$preach_m=p_m preach_{m+1}+(1-p_m) preach_{m-1}$$
 Thus for $0<m<n$,
 $$preach_m-preach_{m+1}=\rho_m(preach_{m-1}-preach_m)$$
 Thus  for $0<m<n$,
 $$preach_m-preach_{m+1}=(preach_{0}-preach_1)\prod_{0<k\leq m}\rho_k$$
 So for $0<m<n$,
 $$preach_0-preach_{m+1}=(preach_{0}-preach_1)\sum_{0\leq j\leq m}\prod_{0<k\leq j}\rho_k$$
 Let $m=n-1$. Then:
$$1=(1-preach_1)\sum_{0\leq j< n}\prod_{0<k\leq j}\rho_k$$
 implying 
 $$1-preach_1=\frac{1}{\sum_{0\leq j\leq n}\prod_{0<k\leq j}\rho_k}$$
 and consequently for
$0\leq m\leq n$,
$$1-preach_{m}=\frac{\sum_{0\leq j< m}\prod_{0<k\leq j}\rho_k}{\sum_{0\leq j< n}\prod_{0<k\leq j}\rho_k}$$
which can be rewritten for $0\leq m\leq n$,
$$preach_{m}=\frac{\sum_{m \leq j< n}\prod_{0<k\leq j}\rho_k}{\sum_{0\leq j< n}\prod_{0<k\leq j}\rho_k}$$

\end{proof}

\begin{proposition}\label{appendix} Let $\mathcal M_1$ be the Markov chain of Figure~\ref{fig:rw}. Then $\mathcal M_1$ is  recurrent if and only if:\\  
$$\sum_{n\in \nat} \prod_{1\leq m<n} \rho_m =\infty\mbox{ with }\rho_m=\frac{1-p_m}{p_m}$$ 
and when transient, the probability that starting from $m$ the random path visits
$0$ is equal to:
$$\frac{\sum_{m\leq n} \prod_{1\leq m<n} \rho_m}{\sum_{n\in \nat} \prod_{1\leq m<n} \rho_m}$$
\end{proposition}
\begin{proof}
Let $0<m$ be some state of $\mathcal M_1$ and $n>m$. Consider the set of paths $\Omega_n$ that, starting from $m$ visit state 0 without having previously visiting state $n$
and $\Omega_\infty$ the set of paths that, starting from $m$ visit state 0. Obviously $(\Omega_n)_{n>m}$ is a non decreasing sequence of sets and $\bigcup_{n>m} \Omega_n=\Omega _\infty$.
Thus $\prob_{\mathcal M_1,m}(\Omega_\infty)=\lim_{n\rightarrow \infty}\prob_{\mathcal M_1,m}(\Omega_n)$.

\smallskip\noindent
Observe now that $\prob_{\mathcal M_1,m}(\Omega_n)$ is the probability in the Markov chain of Figure~\ref{fig:gambler} that,
starting from $m$ the random path reaches $0$, i.e., $\frac{\sum_{m \leq j< n}\prod_{0<k\leq j}\rho_k}{\sum_{0\leq j< n}\prod_{0<k\leq j}\rho_k}$.
Thus $\prob_{\mathcal M_1,m}(\Omega_\infty)=\lim_{n\rightarrow \infty} \frac{\sum_{m \leq j< n}\prod_{0<k\leq j}\rho_k}{\sum_{0\leq j< n}\prod_{0<k\leq j}\rho_k}$.

\smallskip\noindent
This probability is equal to 1 if and only if $\sum_{n\in \nat} \prod_{1\leq m<n} \rho_m =\infty$ which is equivalent to the recurrence of $\mathcal M_1$.\\
If $\mathcal M_1$ is transient then $\prob_{\mathcal M_1,m}(\Omega_\infty)= \frac{\sum_{m \leq j}\prod_{0<k\leq j}\rho_k}{\sum_{ j\in \nat}\prod_{0<k\leq j}\rho_k}$
which concludes the proof.

\end{proof}

\begin{lemma}\label{und-norm}
The halting problem of normalized counter machines is undecidable.
\end{lemma}

\begin{proof}
Let $\mathcal C$ be a two-counter machine with initial values $v_1,v_2$,
 one builds the normalized two-counter machine $\mathcal C_{v_1,v_2}$ by adding after the two first instructions
 of a normalized  two-counter machine, $v_1$ incrementations of $c_1$ followed by $v_2$ incrementations of $c_2$
 followed by  the instructions $\mathcal C$ where the halting instruction has been replaced by  the last three instructions
 of a normalized  two-counter machine. The normalized two-counter machine $\mathcal C_{v_1,v_2}$ halts if and only if 
  $\mathcal C$ with initial values $v_1,v_2$ halts.
  
\end{proof}  

\undecPCM*
\begin{proof}
Let us reduce the halting problem of normalized counter programs
with counters $c_1$, $c_2$ to the decisiveness of pCM.
Let $\mathcal P$ be a normalized counter program  with set of instructions $\{0,\ldots,n\}$. 
One builds $\mathcal C$ as follows.
\begin{itemize}
	\item $Q=\{0,\ldots,n\}$;
	\item $P=\{c_1,c_2\}$;
	\item For all $0\leq i<n$, $i \xrightarrow{(0,0),(1,0)} 0 \in \Delta_1$ with constant weight 4.\\
	All other transitions have constant weight 1;
	\item For all instruction $i:c_j \leftarrow c_j+1; \mathbf{ goto~} i'$,
	$i \xrightarrow{(0,0),(1_{j=1},1_{j=2})} i' \in \Delta_1$;
	\item For all instruction $i:\mathbf{if~} c_j>0 \mathbf{~then~} c_j \leftarrow c_j-1; \mathbf{goto~} i'$, $\mathbf{else\ goto~} i''$,\\
	$i \xrightarrow{(1_{j=1},1_{j=2}),(0,0)} i' \in \Delta_1$ and $i \xrightarrow{c_j,(0,0)} i'' \in \Delta_0$.
\end{itemize}
Let us consider the decisiveness w.r.t. $s_0=(0,0,0)$ (i.e. the initial instruction
with null counters) and  $A=\{(n,0,0)\}$ (i.e. the halt instruction
with null counters). 

\noindent
$\bullet$ Assume that $\mathcal P$ does not halt. Let us recall that in this case it does not
halt for any initial values of the counters. Thus adding the transitions $i \xrightarrow{(0,0),(1,0)} 0$
that ``restarts the counter machine'' does not change the (qualitative) behaviour: the halt instruction is still unreachable
in $\mathcal C$. So $\mathcal C$ is decisive w.r.t. $s_0$ and $A$.

\noindent
$\bullet$ Assume that $\mathcal P$ halts. Let us recall that in this case it halts
for any initial values of the counters. Thus again adding the transitions $i \xrightarrow{(0,0),(1,0)} 0$
that ``restarts the counter machine'' does not change the (qualitative) behaviour: the halt instruction is reachable from any reachable configuration
of $\mathcal C$. Now consider the evolution of the value of counter $c_1$ starting from $s_1=(0,1,0)$ reached with probability at least $\frac 2 3$ after the first instruction.
When its value is greater than 0, it is incremented with probability at least $\frac 2 3$ and  with probability at most $\frac 1 3$ it can be either unchanged,
incremented or decremented. Observe that since $\mathcal P$ is normalized, as long  as $c_1>0$, $\mathcal C$ cannot reach the halt instruction.
Thus the probability to reach the $halt$ instruction starting from $s_1$ is upper bounded by the probability to reach the state $0$
in this random walk:

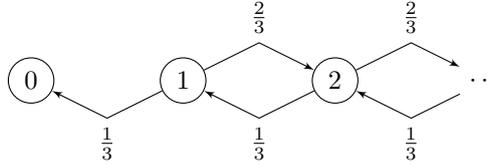
\vspace{-3mm}
\begin{center}
  \begin{tikzpicture}[node distance=2cm,->,auto,-latex]
    
   \path (2,0) node[minimum size=0.6cm,draw,circle,inner sep=2pt] (q1) {$0$};
   \path (4,0) node[minimum size=0.6cm,draw,circle,inner sep=2pt] (q2) {$1$};
   \path (6,0) node[minimum size=0.6cm,draw,circle,inner sep=2pt] (q3) {$2$};
   \path (8,0) node[] (q4) {$\cdots$};

   \draw[arrows=-latex'] (q2)-- (3,-0.5) node[pos=1,below] {$\frac 1 3$}--(q1) ;

    \draw[arrows=-latex'] (q2)-- (5,0.5) node[pos=1,above] {$\frac 2 3$}--(q3) ;
    \draw[arrows=-latex'] (q3)-- (5,-0.5) node[pos=1,below] {$\frac 1 3$}--(q2) ;

    \draw[arrows=-latex'] (q3)-- (7,0.5) node[pos=1,above] {$\frac 2 3$}--(q4) ;
    \draw[arrows=-latex'] (q4)-- (7,-0.5) node[pos=1,below] {$\frac 1 3$}--(q3) ;

  \end{tikzpicture}
\end{center}
\vspace{-3mm}
Due to Proposition~\ref{appendix} in Appendix, this probability is less than 1 and so $\mathcal C$ is not decisive w.r.t. $s_0$ and $A$.

\end{proof}

\decdivpda*
\begin{proof}
W.l.o.g. we assume that $s_0=n_\iota$ and $A=\{n_f\}$ with $n_\iota>n_f$. 
The other cases either reduce to this one or do not present difficulties.

\noindent
Such $\mathcal C$ can be described as follows. First, the transitions of $\Delta_0$ are irrelevant and so
we assume that  $\Delta_0=\emptyset$. Since deleting transition $t$ with$(\mathbf{Pre}(t), \mathbf{Post}(t))=(1,1)$  does not change the behaviour of $\mathcal C$,
we are left with two transitions.
 \begin{itemize}[nosep]
	\item $dec$ with $(\mathbf{Pre}(dec), \mathbf{Post}(dec))=(1,0)$;
	\item $inc$ with $(\mathbf{Pre}(inc), \mathbf{Post}(inc))=(1,2)$.
\end{itemize}
Observe
that except for the absorbing state $0$, $\mathcal M_\mathcal C$
is isomorphic to $\mathcal M_1$. Thus we introduce $\rho_n=\frac{W(dec,n)}{W(inc,n)}$ for $n>0$.

\noindent
Let us write $W(dec,n)=\sum_{i\leq d}a_in^i$ and  $W(inc,n)=\sum_{i\leq d'}a'_in^i$.\\
We perform a case study analysis.  

\smallskip\noindent $\bullet$ When:
\begin{itemize}[nosep] 
	\item $d'<d$
	\item or $d'=d$,  $i_0=\max(i \mid a_i\neq a'_i)$ exists and $a_{i_0}>a'_{i_0}$
	\item or $W(dec)= W(inc)$
\end{itemize}
Then there exists $n_0$ such that for all $n \geq n_0$,
$W(inc,n)\leq W(dec,n)$ implying $\rho_n\geq 1$.\\
Thus for all $n\geq n_0$, $\prod_{1\leq m\leq n}\rho_m\geq \prod_{1\leq m\leq n_0}\rho_m$ 
implying $\sum_{n\in \nat}\prod_{1\leq m\leq n}\rho_m=\infty$ yielding decisiveness.

\smallskip\noindent
$\bullet$ {\bf Case} $d'=d$ and $i_0=\max(i \mid a_i\neq a'_i)$ exists and $a_{i_0}<a'_{i_0}$ and
$i_0\leq d-2$. Then there exists $n_0$ and $\alpha>0$ such that for all $n>n_0$,
$\rho_n\geq 1- \frac{\alpha}{n^2}$.
Observe that:
\begin{align*}
 \prod_{n_0<m\leq n} \rho_m &\geq&  \prod_{n_0<k\leq n} 1- \frac{\alpha}{m^2}&\geq& \prod_{n_0<m\leq n} e^{- \frac{2\alpha}{m^2}}
&=& e^{- \sum_{n_0<m\leq n} \frac{2\alpha}{m^2}}
&\geq& e^{- \sum_{n_0<m} \frac{2\alpha}{m^2}}>0
\end{align*}
Thus $\sum_{n\in \nat}\prod_{1\leq m\leq n}\rho_m=\infty$ yielding decisiveness.

\smallskip\noindent
$\bullet$ {\bf Case} $d'=d$ and $i_0=\max(i \mid a_i\neq a'_i)$ exists, $i_0= d-1$ and $0<\frac{a'_{d-1}-a_{d-1}}{a_d}\leq 1$.\\ 
Let $\alpha=\frac{a'_{d-1}-a_{d-1}}{a_d}$.
Then there exists $n_0$ and $\beta>0$ such that for all $n>n_0$,
$\rho_n\geq 1- \frac{\alpha}{n}-\frac{\beta}{n^2}$.
Observe that:\\
\begin{align*} 
\prod_{n_0<m\leq n} \rho_n&\geq\  \prod_{n_0<m\leq n} 1- \frac{\alpha}{m}- \frac{\beta}{m^2}
\geq\ \prod_{n_0<m\leq n} e^{- \frac{\alpha}{m}- \frac{\beta}{m^2}-(\frac{\alpha}{m}- \frac{\beta}{m^2})^2}
\geq \prod_{n_0<m\leq n} e^{- \frac{\alpha}{m}- \frac{\beta'}{m^2}}\\
\mbox{\emph{for some }}\beta'& \\
&=\ e^{-\sum_{n_0<m\leq n}\frac{\alpha}{m}+ \frac{\beta'}{m^2}} 
\geq\ e^{-\sum_{n_0<m} \frac{\beta'}{m^2}}e^{- \alpha\log(n)}=\frac{e^{-\sum_{n_0<m} \frac{\beta'}{m^2}}}{n^\alpha}
\end{align*} 
Thus $\sum_{n\in \nat}\prod_{1\leq m\leq n}\rho_m=\infty$ yielding decisiveness.

\smallskip\noindent
$\bullet$ {\bf Case} $d'=d$ and $i_0=\max(i \mid a_i\neq a'_i)$ exists, $i_0= d-1$ and $\frac{a'_{d-1}-a_{d-1}}{a_d}> 1$.\\ 
Let $\alpha=\frac{a'_{d-1}-a_{d-1}}{a_d}$.
Then there exists $n_0$ and $1<\alpha'<\alpha$ such that
for all $n\geq n_0$,  $\rho_n\leq  1- \frac{\alpha'}{n}$.
Observe that:
\begin{align*}
 \prod_{n_0<m\leq n} \rho_m&\leq\  \prod_{n_0<m\leq n} 1- \frac{\alpha'}{m}
 &\leq\ \prod_{n_0<m\leq n} e^{- \frac{\alpha'}{m}}
&=\ e^{-\sum_{n_0<m\leq n}\frac{\alpha'}{m}}\\
&\leq\ e^{- \alpha'(\log(n+1)-\log(n_0+1))}
&=\ \frac{e^{ \alpha'\log(n_0+1)}}{(n+1)^{\alpha'}}
\end{align*}
Thus $\sum_{n\in \nat}\prod_{1\leq m\leq n}\rho_m<\infty$ implying non decisiveness.

\smallskip\noindent
$\bullet$  When (1) $d=d'$ and $a'_d>a_d$ or (2) $d<d'$.\\
Then there exists $n_0$ and $\alpha<1$ such that
for all $n\geq n_0$,  $\rho_n \leq \alpha$.\\
Thus $\sum_{n\in \nat}\prod_{1\leq m\leq n}\rho_m<\infty$ implying non decisiveness.

\smallskip\noindent
This concludes the proof that the decisiveness is decidable for $\mathcal C$. 

\end{proof}

\qualpHM*
\begin{proof}
We first observe that if  $Q\times \{0\}$ is reachable from $(q,k)$
then it is reachable from $(q,k')$ for all $k'\leq k$.
Let $Reach_n$ be the set of configurations whose length of a shortest path 
to $Q \times \{0\}$ is $n$ 
. Consider a configuration $(q,k) \in Reach_{n}$
with $n\geq |Q|$.
A shortest path from $(q,k)$ to $Q \times \{0\}$ visits twice some state $q'$
and thus this path includes a circuit of  transitions in $\Delta_1$, $q_0\xrightarrow{1,v_1}q_1\cdots \xrightarrow{1,v_m} q_m$
with $q_0=q_m=q'$ and $\sum_{i\leq m} v_i<m$. We call such a circuit, a \emph{negative circuit}.

\noindent
$\bullet$ Let $q\in Q$. Assume that starting from $q$, there is a sequence of  transitions in $\Delta_1$
reaching a negative circuit. Observe that  $Q\times \{0\}$ is reachable from any $(q,k)$
by iterating the circuit.
On the other hand, the length of such a shortest sequence is at most $|Q|-1$
and the length of the circuit is at most $|Q|$. Thus $(q,|Q|)$ belongs to $Reach_n$
for some $n \in [|Q|,(2|Q|-1)|Q|]$.

\noindent
$\bullet$ Let $q\in Q$. Assume that starting from $q$, there is no sequence of  transitions in $\Delta_1$
reaching a negative circuit. So no $(q,k)$ may appear in $Reach_n$ for $n\geq |Q|$.

\noindent
Thus the algorithm proceeds as follows. It computes incrementally $Reach_n$
for $n\leq (2|Q|-1)|Q|$. Let $q\in Q$ and $s_q$ be the maximum $k$ 
such that $(q,k)$ occurs in some $Reach_{n_q}$. If $n_q<|Q|$ then $r_q=s_q$, 
otherwise $r_q=\infty$.

\end{proof}


\end{document}